\documentclass[11pt]{article}

\usepackage[utf8]{inputenc}
\usepackage[T1]{fontenc}
\usepackage[english]{babel}
\usepackage{amsmath,amsthm,amssymb}
\usepackage{thm-restate}
\usepackage{algorithm}
\usepackage{graphicx}
\usepackage{caption}
\usepackage{subcaption}
\usepackage[colorlinks=true,linkcolor=blue,citecolor=blue,urlcolor=blue]{hyperref}
\usepackage[capitalize,noabbrev]{cleveref}
\usepackage{geometry}
\pdfoutput=1

\newtheorem{theorem}{Theorem}[section]
\newtheorem{lemma}[theorem]{Lemma}
\newtheorem{corollary}[theorem]{Corollary}
\newtheorem*{claim*}{Claim}

\newenvironment{steplist}{\begin{enumerate}}{\end{enumerate}}

\renewcommand{\epsilon}{\varepsilon}
\renewcommand{\phi}{\varphi}
\DeclareMathOperator{\symdiff}{\triangle}

\usepackage{refcount}

\title{Randomized Near-Linear Time Computation of Welzl Orders on Graphs with Linear Neighborhood Complexity}

\author{
Jan Dreier\\
TU Wien, Austria\\
\texttt{dreier@ac.tuwien.ac.at}
\and
Clemens Kuske\\
TU Wien, Austria\\
\texttt{mail@clemens-kuske.de}
}

\date{}

\begin{document}
\maketitle

\begin{abstract}
Orders with low crossing number, introduced by Welzl, are a fundamental tool in range searching and computational geometry. Recently, they have found important applications in structural graph theory: set systems with linear shatter functions correspond to graph classes with linear neighborhood complexity.
For such set systems $\mathcal{S} = (U, \mathcal{F})$, Welzl's theorem guarantees the existence of orders with only $\mathcal{O}(\log^2 n)$ crossings, $n = |U|$.
A series of works has progressively improved the runtime for computing such orders, from Chazelle and Welzl's original $\mathcal{O}(|U|^3 |\mathcal{F}|)$ bound~\cite{CW89}, through Har-Peled's $\mathcal{O}(|U|^2|\mathcal{F}|)$~\cite{HP09}, to the recent sampling-based methods of Csikós and Mustafa~\cite{CM24}.

We present a randomized algorithm that computes Welzl orders for set systems with linear primal and dual shatter functions in expected time $\mathcal{O}(\|S\| \log \|S\|)$, where $\|S\| = |U| + \sum_{X \in \mathcal{F}} |X|$ is the size of the canonical input representation. The algorithm achieves an asymptotically optimal $\mathcal{O}(\log n)$ crossing bound.
As an application, we compute compact neighborhood covers in graph classes with (near-)linear neighborhood complexity in time \(\mathcal{O}((n+m) \log n)\), where \(n\) and \(m\) are the numbers of vertices and edges. Our result also improves the runtime of first-order model checking on monadically stable graph classes from $\mathcal{O}(n^{5+\varepsilon})$ to $\mathcal{O}(n^{3+\varepsilon})$.
\end{abstract}

\section{Introduction}

Orders with low crossing number, first introduced by Welzl~\cite{Wel88,Wel92} in the area of range searching and combinatorial geometry, have proven to be a remarkably versatile tool across many algorithmic domains. Given a set system $\mathcal{S} = (U, \mathcal{F})$ consisting of a ground set $U$ and a family of subsets $\mathcal{F}$ over the ground set, the \emph{crossing number} of a total order on $U$ with respect to a set $X \in \mathcal{F}$ counts the number of adjacent pairs $(u, u')$ in the order such that exactly one of $u$ and $u'$ belongs to $X$. The crossing number of the order is the maximum crossing number over all sets in $\mathcal{F}$. We call an order achieving a small crossing number a \emph{Welzl order}.

More formally, we define the crossing number as follows. Let $\preccurlyeq$ be a linear order on the ground set $U$. For a set $F \in \mathcal{F}$, the \emph{crossing number} of $F$ in $\preccurlyeq$ is the number of pairs $(u, u')$ such that $u'$ is the immediate successor of $u$ in $\preccurlyeq$, and $u \in F$ if and only if $u' \not\in F$.
The \emph{crossing number} of $\preccurlyeq$ is the maximum crossing number over all sets $F \in \mathcal{F}$.

The existence and computability of Welzl orders is intimately connected to Vapnik-Chervonenkis (VC) theory.
The \emph{shatter function} $\pi_\mathcal{F}(n)$ of a set system $(U, \mathcal{F})$ is defined as \[
\pi_\mathcal{F}(n) := \max_{A \subseteq U, |A| \leq n} |\{X \cap A : X \in \mathcal{F}\}|,
\]
measuring the maximum number of distinct traces that sets from $\mathcal{F}$ can leave on subsets of size~$n$.
Denote by \(\pi^*_\mathcal{F}(n)\) the shatter function of the dual set system. The dual of a set system is the set system obtained by switching the roles of ground set and set family. Formally, the dual of $(U, \mathcal{F})$ is $(\mathcal{F}, \{X^*_{u} \mid u \in U\})$, where $X^*_{u} = \{X \in \mathcal{F} \mid u \in X\}$.
The foundational connection between shatter functions and low-crossing orders is:
\begin{theorem}[{\cite{Wel88,CW89,Wel92}}]\label{thm:welzl}
Let $\mathcal{S} = (U, \mathcal{F})$ be a set system with dual shatter
function \(\pi^*_\mathcal{F}(k) = \mathcal{O}(k^d)\) for some \(d \ge 1\).
Then there exists a total order on $U$ with crossing number $\mathcal{O}(|U|^{1-1/d} \cdot \log |U|)$ if \(d > 1\),
and \(\mathcal{O}(\log^2 |U|)\) if \(d=1\).
\end{theorem}
For comparison, a set of size $n$ can be shattered into at most $2^n$ distinct traces. Any order has a crossing number of at most $|U| - 1$, as that is the number of adjacent pairs in the order. Further there are set systems with $n = |U| = |\mathcal{F}|$ and a crossing number of $\sim \tfrac{n}{2}$, for example hypercube cuts.

\paragraph*{Prior work}
There is a large body of work dedicated to finding fast algorithms for computing Welzl orders\footnote{Some prior work computes matchings or trees instead of orders, but these structures behave similarly~\cite{Wel92}.} satisfying the guarantees of \Cref{thm:welzl}.
Chazelle and Welzl's original algorithm~\cite{CW89,Wel92} is deterministic, uses an iterative weight-doubling technique, and has a runtime of $\mathcal{O}(|U|^3 |\mathcal{F}|)$.
Har-Peled~\cite{HP09} presents an alternative approach based on linear programming and randomized rounding that, with state-of-the-art LP solvers, further reduces the runtime to $\mathcal{O}(|U|^2 |\mathcal{F}|)$ (see also ~\cite{FedLub04}).
In the research direction of deterministic algorithms, Ducoffe, Habib, and Viennot~\cite{DHV19} gave an algorithm for orders with sublinear crossing number on bounded-VC-dimension set systems in time $\tilde{\mathcal{O}}(\sum_{X \in \mathcal{F}} |X| + |U|^{2-\varepsilon})$ for some $\varepsilon > 0$ depending on the VC-dimension.
Recently, Csikós and Mustafa~\cite{CM21} used sampling-based methods
to obtain a runtime of $\tilde{\mathcal{O}}\big(|\mathcal{F}| \cdot |U|^{2/d} + |U|^{2+2/d}\big)$, improving previous runtimes for the case of \(d > 1\).
In a follow-up work~\cite{CM24}, they further obtain a smooth tradeoff:
for any $\alpha \in (0,1]$,
one can compute an order with expected crossing number
$\mathcal{O}(|U|^{1-\alpha/d} + \log |\mathcal{F}| \log |U|)$
in expected time $\widetilde{\mathcal{O}}(|U|^{1+\alpha + 2\alpha/d} + |\mathcal{F}| \cdot |U|^{2\alpha/d})$.
As $\alpha$ approaches zero, the runtime approaches near-linear, but so does the crossing number
(forfeiting the $\mathcal{O}(|U|^{1-1/d})$ guarantee that makes Welzl orders useful).
In particular, for $d = 1$, the best previously stated runtime bound for achieving
$\mathcal{O}(\log^2 |U|)$ crossings was \(\mathcal{O}(|U|^2 |\mathcal{F}|)\).

Another line of work studies total crossings, for example for diameter
algorithms~\cite{CCGKLZ25,DKP24} and matrix--vector multiplication~\cite{ABM25}.
While the goals sound similar, our local requirement that each set have a
small number of crossings requires a much more refined approach to reach
the same bounds.

\paragraph*{Our contribution}
Recently, Welzl orders have been proven to be highly useful in the context of structural graph theory.
Here, set systems with linear (primal and dual) shatter functions are of particular interest, in particular the case \(d=1\), where \Cref{thm:welzl} promises the existence of orders with $\mathcal{O}(\log^2 n)$ crossings. Note that the bound on both shatter functions is an additional restriction, not present in earlier work. In our graph theoretical setting, however, this restriction is essentially free, since the neighborhood set system is self-dual, as explained below.

Denote by \(\|S\| := |U| + \sum_{X \in \mathcal{F}} |X|\) the size of the canonical representation of the input set system.
This is the total incidence size of the input; in particular, \(\|S\| \le |U| + |U|\,|\mathcal{F}|\).

We present the algorithm in two stages. First, a simple construction based on uniform sampling and twin partitions achieves the classical $\mathcal{O}(\log^2 |U|)$ crossing bound in near-linear time:

\begin{restatable}{theorem}{mainresult}\label{thm:main}
Let $\mathcal{S} = (U, \mathcal{F})$ be a set system with primal and dual shatter
functions \(\pi_\mathcal{F}(k),\pi^*_\mathcal{F}(k) \le c \cdot k\) for some \(c \ge 1\).
There is a randomized algorithm that
takes \(\mathcal{S}\) and \(c\) as input and computes
in time \(\mathcal{O}(\|S\|\cdot \log \|S\|)\),
with probability \(2/3\),
a total order on \(U\) with crossing number at most $12c^2 \log^2 |U|$.
\end{restatable}

The previous best runtime in this regime was the also randomized \(\mathcal{O}(|U|^2 |\mathcal{F}|)\) bound of Har-Peled~\cite{HP09}.
Note that it takes at least time \(\Omega(\|S\|)\) to read the whole input instance, so our runtime is within a logarithmic factor of the theoretical minimum. By standard boosting techniques, the failure probability can be made arbitrarily small.
If $c$ is unknown, we can run the boosted algorithm with increasing guesses for $c$ until it succeeds. The first successful run yields a valid order. With doubling guesses, this increases the runtime by a factor of $\mathcal{O}(\log c)$.

In the second stage, we refine the sample to improve the crossing number to $\mathcal{O}_c(\log |U|)$ while retaining near-linear expected time:

\begin{theorem}\label{thm:refined-main}
Let $\mathcal S=(U,\mathcal F)$ be a set system with primal and dual shatter functions
$\pi_{\mathcal F}(k),\pi^*_{\mathcal F}(k)\leq ck$ for all $k\geq1$, where $c\geq2$.
There is a randomized algorithm that takes $\mathcal S$ and $c$ as input and computes a total order on $U$ with crossing number
$\mathcal O(c^5\log^2 c\,\log|U|)$ in expected time
$\mathcal O\bigl(\|S\|\log\|S\|\bigr)$.
\end{theorem}

The logarithmic dependence of the crossing bound on $|U|$ is asymptotically optimal in general, even when $c$ is a fixed constant.
Indeed, cographs have bounded twin-width and hence linear neighborhood complexity~\cite{BFLP24}, and Crespelle and Gambette~\cite[Theorem~3]{CG14} exhibit cographs on $n$ vertices for which every vertex order represents some closed neighborhood as a union of $\Omega(\log n)$ maximal intervals.
A set occupying $r$ maximal intervals has at least $2r-2$ crossings, so these graphs require $\Omega(\log n)$ crossings.
The expected runtime remains within a logarithmic factor of the input size.

\paragraph*{Applications in structural graph theory}
For a graph \(G\), the closed \(r\)-neighborhood of a vertex \(v\) consists of all vertices at distance at most \(r\) from \(v\), including \(v\) itself; we write \(N[v]\) for the case \(r=1\), so \(N[v]=\{v\}\cup\{u:uv\in E(G)\}\).
The \emph{neighborhood set system} of $G$ is defined as
$(V(G), \{ N[v] \mid v \in V(G) \})$.
The \emph{neighborhood complexity} $\pi_G$ of $G$ is the shatter function of this set system.
Here and throughout, neighborhood complexity refers to closed neighborhoods, that is, to the case \(r=1\). This is the precise sense in which set systems with linear shatter functions correspond to graph classes with linear neighborhood complexity: every graph \(G\) yields a set system whose primal shatter function is \(\pi_G\), and since \(G\) is undirected its dual shatter function coincides with \(\pi_G\) as well; conversely, every set system with linear primal and dual shatter functions can be encoded by a bipartite incidence graph with linear neighborhood complexity.

We further define for a graph \(G\) the \emph{crossing number} of a total order on $V(G)$
as its crossing number with respect to the neighborhood set system.
With this notation (and using the fact that the neighborhood set system equals its own dual)
\Cref{thm:main} yields:

\begin{theorem}\label{thm:main_graph}
Let $G$ be a graph with \(n\) vertices, \(m\) edges, and neighborhood complexity $\pi_G(k) \le c \cdot k$ for some \(c \ge 1\).
There is a randomized algorithm that
takes \(G\) and \(c\) as input and
computes
in time \(\mathcal{O}\bigl((n + m) \log n\bigr)\),
with probability \(2/3\),
a total order on $V(G)$ with crossing number at most $12c^2 \log^2 n$.
\end{theorem}

This improves upon the previously best runtime of \(\mathcal{O}(n^3)\) for the linear case implied by Har-Peled's result~\cite{HP09}.
Applying \Cref{thm:refined-main} to the same neighborhood set system further gives an order with $\mathcal O_c(\log n)$ crossings in expected time $\mathcal O((n+m)\log n)$.

For a graph class $\mathcal{C}$, define its neighborhood complexity as
$\pi_{\mathcal{C}}(k) = \max_{G \in \mathcal{C}} \pi_G(k)$.
Many natural graph classes have linear neighborhood complexity,
including planar graphs, graphs of bounded genus,
graphs excluding a fixed minor or topological minor,
and more generally all graph classes of bounded expansion~\cite{RVS19},
as well as graph classes of bounded twin-width~\cite{BFLP24}.
For such classes, the refined algorithm of \Cref{thm:refined-main} computes Welzl orders
with $\mathcal{O}_{\mathcal{C}}(\log n)$ crossings in expected time $\mathcal O((n+m)\log n)$.
Several important graph classes have \emph{almost-linear} neighborhood complexity,
meaning $\pi_{\mathcal{C}}(k) = \mathcal{O}_\varepsilon(k^{1+\varepsilon})$ for every $\varepsilon > 0$.
This includes nowhere dense graph classes~\cite{EGKKPRS17},
monadically stable graph classes~\cite{ms},
and all monadically dependent graph classes~\cite{DMMPT26}.
For such classes, \Cref{thm:main_graph} yields Welzl orders
with $\mathcal{O}_{\mathcal{C},\varepsilon}(n^\varepsilon)$ crossings for every $\varepsilon > 0$.

Let us discuss one application of \Cref{thm:main_graph}.

A \emph{compact neighborhood cover} of a graph \(G\) is a family
\(\mathcal{K}\) of subsets of vertices of \(G\)
such that for every vertex \(v\) of \(G\), there exists $C \in \mathcal K$ such that $N[v] \subseteq C$,
and, moreover, each set \(C \in \mathcal{K}\) has (strong) diameter at most 4.
The \emph{overlap} of \(\mathcal K\) is
\(\max_{v \in V(G)} |\{ C \in \mathcal K \mid v \in C \}|\).
Compact neighborhood covers with small overlap play a key role in first-order model checking algorithms~\cite{nd,snd,ms}.
The following is a direct consequence of \Cref{thm:main_graph} and the algorithm given in the proof of~\cite[Lemma 17]{ms}:

\begin{theorem}\label{thm:neighborhoodCover}
    Let $G$ be a graph with \(n\) vertices, \(m\) edges, and neighborhood complexity $\pi_G(k) \le c \cdot k$ for some \(c \ge 1\).
    There is a randomized algorithm that
    takes \(G\) and \(c\) as input and
    computes
    in time $\mathcal{O}\bigl((n + m) \log n\bigr)$,
    with probability \(2/3\),
    a compact neighborhood cover with overlap at most $1+12c^2 \cdot \log^2 n$.

\end{theorem}

Compact neighborhood covers have been a central ingredient in model checking algorithms for first-order logic on well-structured graph classes~\cite{nd,snd,ms}.
For nowhere dense graph classes~\cite{nd}, such covers were computed in near-linear time using generalized coloring numbers.
For the more general structurally nowhere dense graph classes~\cite{snd}, this approach was infeasible, and an LP formulation was used to compute compact neighborhood covers in time \(\mathcal{O}(n^{9.8})\).
Then for monadically stable classes~\cite{ms}, the connection between compact neighborhood covers and Welzl orders was discovered,
and Welzl's original algorithm was used to compute neighborhood covers in time roughly \(\mathcal{O}(n^{4})\),
leading to a total runtime of $\mathcal{O}_{\mathcal{C},\epsilon,\varphi}(n^{5+\varepsilon})$ for model checking (subscripts in the \(\mathcal{O}\)-notation hide constant factors depending on the subscripts).
While not mentioned in the paper, Har-Peled's \(\mathcal{O}(n^{3})\)-algorithm~\cite{HP09} can be used to further reduce to $\mathcal{O}_{\mathcal{C},\epsilon,\varphi}(n^{4+\varepsilon})$.
By further substituting our near-linear time \Cref{thm:neighborhoodCover} we obtain:

\begin{corollary}[based on {\cite{ms}}]\label{thm:monstable}
Let $\mathcal{C}$ be a monadically stable graph class and $\varepsilon > 0$.
There is a randomized algorithm that, given an $n$-vertex graph $G \in \mathcal{C}$
and a first-order sentence $\varphi$,
decides whether $\varphi$ holds in $G$ in time $\mathcal{O}_{\mathcal{C},\varepsilon,\varphi}(n^{3 +\varepsilon})$, with high probability.
\end{corollary}

By standard boosting, we can make the failure probability of \Cref{thm:neighborhoodCover} subpolynomially small, while increasing the runtime only by a logarithmic factor. A union bound over the polynomially many invocations then keeps the overall failure probability subpolynomial.
With this in mind, the proof of \Cref{thm:monstable} closely follows the arguments of~\cite{ms}: one first improves the runtime of their Theorem 14 from \(\mathcal{O}_{\mathcal{C},r,\epsilon}(n^{4+\epsilon})\) to \(\mathcal{O}(n^2 \log n)\) using our \Cref{thm:neighborhoodCover} and the ideas outlined in the beginning of~\cite[Section 3.2]{ms}.
Then one follows the involved arguments in~\cite[Appendix B]{ms}, but substitutes their Theorem 14 with our improvement.

\section{Preliminaries}\label{sec:prelims}

Define the \emph{linearity} of a set system \((U,\mathcal F)\) as the infimum \(c \ge 1\) such that
\(\pi_\mathcal{F}(k),\pi^*_\mathcal{F}(k) \le c \cdot k\) for all \(k\).
We write $\log$ for base-\(2\) logarithms and $\ln$ for natural logarithms.
For convenience, we represent set systems \((U,\mathcal{F})\) as bipartite graphs \((A,B,E)\),
with \(A=U\), \(B=\mathcal{F}\) and \(E = \{ uX \mid u \in X, X \in \mathcal{F} \}\).
Here \(N(a)=\{b\in B:ab\in E\}\) and \(N(b)=\{a\in A:ab\in E\}\) denote incidence sets on the opposite side; the closed-neighborhood convention above refers to the original graphs in the applications.
For a bipartite graph \(G\) and \(A' \subseteq A\), \(B' \subseteq B\) we denote by \(G[A',B']\) the induced subgraph on \(A',B'\).

\paragraph*{Twin partitions}

Consider a bipartite graph \(G=(A,B,E)\).
Two vertices \(v,v' \in A\) are \emph{twins} if \(N(v)=N(v')\).
Note that the twin relation is an equivalence relation. We call the (unique) partition of \(A\) into twin equivalence classes the \emph{twin partition} of \(A\).
The definition of twin classes provides a natural way to express the notion of neighborhood complexity.
\[
\pi_G(n) := \max_{A' \subseteq A, |A'| \leq n} |\mathcal{B}_{A'}| \textnormal{\quad\quad where \(\mathcal{B}_{A'}\) is the twin partition of \(B\) in \(G[A',B]\)}.
\]
\[
\pi^*_G(n) := \max_{B' \subseteq B, |B'| \leq n} |\mathcal{A}_{B'}| \textnormal{\quad\quad where \(\mathcal{A}_{B'}\) is the twin partition of \(A\) in \(G[A,B']\)}.
\]
Use \(X \symdiff Y\) to denote the symmetric difference between two sets \(X\) and \(Y\).
We say \(v,v' \in A\) are \emph{\(k\)-near twins}, for \(k \in \mathbb{N}\), if \(|N(v) \symdiff N(v')| \le k\).
For a bipartite graph $G = (A, B, E)$, a partition \(\mathcal{P}=\{P_1, \dots, P_\ell\}\) of the vertex set $A$,
and a set \(P' \subseteq A\) of representatives with \(P'=\{v_1,\dots,v_\ell\}\),
we say that \(\mathcal{P}\) is a \emph{\(k\)-near twin partition with representatives \(P'\)} if for all \(i\) we have \(v_i \in P_i\) and every \(v \in P_i\) is a \(k\)-near twin of \(v_i\).
Note that \(0\)-near twins are twins.

\paragraph*{Relation between twin partitions and Welzl orders}

Before introducing the algorithm, we examine the relation between twin partitions and crossing numbers, as this connection is central to our approach.

The first observation, already made by Welzl~\cite{Wel88}, can be restated as follows: when computing an order on $A$ with low crossing number, we can temporarily ignore the twins in $A$ without increasing the crossing number.

\begin{lemma}[\cite{Wel88}]\label{lemma:suspend-twins}
Let $G = (A, B, E)$ be a set system.
For any order on $A$, duplicating a vertex (that is, inserting a new vertex \(a'\) with \(N(a)=N(a')\) immediately after an existing vertex \(a\) in the order) does not change its crossing number.
\end{lemma}
\begin{proof}
Consider an order $\prec$ on $A$ and a vertex $a$ with successor \(a''\).
We now insert a new vertex $a'$ with $N(a') = N(a)$ immediately after $a$, but before \(a''\).
Fix any \(b \in B\).
The new pair $(a, a')$ contributes no crossing since $a \in N(b) \iff a' \in N(b)$.
Moreover, if a new crossing \((a',a'')\) appears, then previously a crossing \((a,a'')\) existed, which now vanishes.
Hence, the crossing number is unchanged.
\end{proof}

A subroutine of our algorithm contracts the twin partition of $A$ into a set of representatives (while remembering for every removed vertex its corresponding representative).
\Cref{lemma:suspend-twins} ensures that all removed vertices can later be reinserted next to their representative in the final order without increasing the crossing number.

The second observation deals with \emph{near twins} in $B$. When computing an order on $A$, we can replace a $k$-near twin partition of $B$ by its representatives, increasing the crossing number by at most $2k$.

\begin{lemma}\label{lemma:2k-increment}
Let $G = (A, B, E)$ be a set system and let $\mathcal{B}$ be a $k$-near twin partition of $B$ with representatives $B'$.
Any order on \(A\) with crossing number at most \(m\) on \(G[A,B']\) has crossing number at most \(m+2k\) on \(G\).
\end{lemma}
\begin{proof}
Let $\prec$ be an order on $A$ with crossing number at most $m$ on $G[A,B']$.
For any $b \in B$, let $b_{\mathrm{rep}} \in B'$ be its representative in $\mathcal{B}$.
By definition, $|N(b) \symdiff N(b_{\mathrm{rep}})| \le k$.

The crossing number of $\prec$ with respect to $b$ counts adjacent pairs $(a, a')$ where exactly one of $a, a'$ is in $N(b)$.
An adjacent pair $(a, a')$ contributes differently to the crossing numbers for $b$ and $b_{\mathrm{rep}}$
only if at least one of $a, a'$ lies in $N(b) \symdiff N(b_{\mathrm{rep}})$.
Since each vertex participates in at most two adjacent pairs (with its predecessor and successor),
the number of pairs that contribute differently is at most $2k$.
Thus the crossing number for $b$ differs from that for $b_{\mathrm{rep}}$ by at most $2k$.

Since $b_{\mathrm{rep}} \in B'$ and the order has crossing number at most $m$ on $G[A,B']$,
the crossing number for $b$ is at most $m + 2k$.
This holds for all $b \in B$, proving the claim.
\end{proof}

\section{Introducing the algorithm}\label{section:introduction_algorithm_graph_notation}

We now present the algorithm and restate our main result for convenience.
In the bipartite representation introduced in \Cref{sec:prelims}, the algorithm orders the color class \(A=U\).
Afterwards, we introduce a refined sampling approach in \Cref{sec:refined-sampling} that improves the crossing-number bound.
Throughout the algorithmic discussion, instances with $|A|\leq 1$ or
$B=\emptyset$ are handled directly before the main loop; logarithmic
bounds are analyzed only for the nontrivial case $|A|\geq2$ and
$B\neq\emptyset$.

\mainresult*

We will show that \Cref{algorithm:welzl-orders} has these properties. See \Cref{figure:algorithm_illustration} for an illustration of one iteration of \Cref{algorithm:welzl-orders}.

\begin{figure}[htbp]
\captionsetup{type=algorithm}
\centering
\fbox{%
\parbox{0.95\linewidth}{
\textbf{Input.} \(c \ge 1\) and a set system $G=(A,B,E)$ of linearity~$c$.

\textbf{Output.} An order $\prec$ on $A$, or \emph{false}.

\begin{steplist}
  \item \textbf{Initialization.}
    If $|A|\leq 1$ or $B=\emptyset$, return the unique or an arbitrary order on $A$, respectively.
    Set $A_{\mathrm{cur}}\!\gets\!A$, $B_{\mathrm{cur}}\!\gets\!B$, and $L\!\gets\!\emptyset$ (stack of stored pairs).

  \item \textbf{Main loop.}
    While $|A_{\mathrm{cur}}|>12c^2 \cdot \log |A|$:
    \begin{steplist}
	        \item \textbf{Sampling.}
	            Sample $W \subseteq A_{\mathrm{cur}}$ of size
	            $\left\lceil |A_{\mathrm{cur}}|/(2c^2) \right\rceil$ uniformly at random.
        \item \textbf{Partitioning.}
            Let \(\mathcal B\) be the twin partition of \(B_{\mathrm{cur}}\) over \(W\), with representatives \(B'\).
            Let \(\mathcal A\) be the twin partition of \(A_{\mathrm{cur}}\) over \(B'\), with representatives \(A'\).

        \item \textbf{Checking guarantees.}
            Check if \(\mathcal B\) is a \((6c^2 \log |A|)\)-near twin partition of \(B_{\mathrm{cur}}\) in \(G[A_{\mathrm{cur}},B_{\mathrm{cur}}]\) with representatives \(B'\).

            If not, return \emph{false}.

      \item \textbf{Bookkeeping.}
          For each vertex \(v \in A_{\mathrm{cur}} \setminus A'\) with representative \(v_{\mathrm{rep}} \in A'\),
          push the tuple \((v,v_{\mathrm{rep}})\) onto the stack \(L\).
       \item $(A_{\mathrm{cur}},B_{\mathrm{cur}})\gets(A',B')$.
    \end{steplist}

  \item \textbf{Base construction.}\label{algorithm-step:base_construction}
    Choose an arbitrary order $\prec$ on the remaining vertex set $A_{\mathrm{cur}}$.

  \item \textbf{Reconstruction phase.}\label{algorithm-step:reconstruction}
    Pop pairs $(v,v_{\mathrm{rep}})$ from $L$ in reverse insertion order and insert each $v$ directly after $v_{\mathrm{rep}}$ in $\prec$.

  \item \textbf{Output.} Return the constructed order $\prec$.
\end{steplist}
}}
\caption{The main algorithm mentioned in \Cref{thm:main}.}
\label{algorithm:welzl-orders}
\end{figure}

\begin{figure}[t]

\captionsetup[subfigure]{labelformat=empty}
  \centering
  \begin{subfigure}[t]{0.328\linewidth}
    \centering
    \includegraphics[page=1, width=0.8\linewidth]{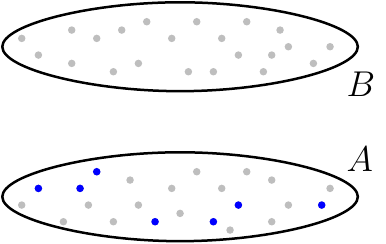}
    \caption{Choose a sample $W \subseteq A$,}
  \end{subfigure}\hfill
  \begin{subfigure}[t]{0.328\linewidth}
    \centering
    \includegraphics[page=2, width=0.8\linewidth]{algorithm-images.pdf}
    \caption{partition $B$ with respect to $W$,}
  \end{subfigure}\hfill
  \begin{subfigure}[t]{0.328\linewidth}
    \centering
    \includegraphics[page=3, width=0.8\linewidth]{algorithm-images.pdf}
    \caption{only keep representatives $B'$,}
  \end{subfigure}\hfill
  \par\medskip
  \begin{subfigure}[t]{0.328\linewidth}
    \centering
    \includegraphics[page=4, width=0.8\linewidth]{algorithm-images.pdf}
    \caption{partition $A$ with respect to $B'$,}
  \end{subfigure}
  \begin{subfigure}[t]{0.328\linewidth}
    \centering
    \includegraphics[page=5, width=0.8\linewidth]{algorithm-images.pdf}
    \caption{store tuples from $A$ to the stack $L$,}
  \end{subfigure}
  \begin{subfigure}[t]{0.328\linewidth}
    \centering
    \includegraphics[page=6, width=0.8\linewidth]{algorithm-images.pdf}
    \caption{and proceed with $G[A', B']$.}
  \end{subfigure}

  \caption{Illustration of one iteration of the algorithm. With $A = A_{\mathrm{cur}}$ being the lower and $B = B_{\mathrm{cur}}$ the upper color class}
  
  \label{figure:algorithm_illustration}
\end{figure}

The following subsections bound the number of iterations, the runtime, the crossing number, and the failure probability of \Cref{algorithm:welzl-orders}.
With these in place, the proof of \Cref{thm:main} follows.

\begin{proof}[Proof of \Cref{thm:main}]
By \Cref{theorem:runtime-bound} (proved in \Cref{runtime}), the algorithm terminates in time $\mathcal{O}\bigl((|A| + |E|)\log |A|\bigr)$.
By \Cref{theorem:resulting-crossing-number} (proved in \Cref{sec:crossing-number}), \Cref{algorithm:welzl-orders} either returns \emph{false} or computes an order on $A$ with crossing number at most $12c^2 \log^2 |A|$.
By \Cref{theorem:probability_of_error} (proved in \Cref{section:probability_of_bad_run}), the probability of returning \emph{false} is at most $1/3$.
\end{proof}

\subsection{Bounding the number of iterations}

\begin{theorem}\label{theorem:iteration_bound}
    Given a set system $G = (A, B, E)$ of linearity~$c \ge 1$,
    the main loop of \Cref{algorithm:welzl-orders} terminates after at most $\log |A|-1$ iterations.
\end{theorem}
\begin{proof}
    Let \(A_i\) and \(B_i\) be the values of \(A_{\mathrm{cur}}\) and \(B_{\mathrm{cur}}\) at the beginning of the \(i\)th iteration of the main loop.
    Thus, \(A_1 = A\) and \(B_1=B\).
    We argue that in the \(i\)th iteration, \(|A_{i + 1}| \leq 1/2 \cdot |A_i| + c^2\).

    Using the size bound $|W_i| = \left\lceil |A_i|/(2c^2) \right\rceil$ for the set \(W_i\) sampled in the \(i\)th iteration,
    and the fact that \(G\) has linearity \(c\), we obtain
    \[
        | A_{i + 1} |
        \le c \cdot | B_{i + 1} |
        \le c^2 \cdot | W_i |
        = c^2 \cdot \left\lceil \frac{|A_i|}{2c^2} \right\rceil
        < c^2 \left(\frac{|A_i|}{2c^2} + 1\right)
        = \frac{1}{2}|A_i| + c^2.
    \]

    We now unfold the recurrence.  By induction,
    \[
        |A_{i+1}|
        \le
        \frac{|A|}{2^i}
        +
        c^2\Bigl(1 + \tfrac12 + \cdots + \tfrac{1}{2^{i-1}}\Bigr).
    \]
    The geometric sum is bounded by $2$, and, hence,
    \[|A_{i+1}| \le \frac{|A|}{2^i} + 2c^2.\]

    The main loop terminates after iteration \(i\) if $|A_{i+1}| \le 12c^2 \cdot \log |A|$.  Therefore, it is sufficient
    to find $i$ such that
    \[\frac{|A|}{2^i} + 2c^2 \le 12c^2
        \quad\Longleftrightarrow\quad
        \frac{|A|}{2^i} \le 10c^2
        \quad\Longleftrightarrow\quad
    2^i \ge \frac{|A|}{10c^2} \quad\Longleftrightarrow\quad i \ge \log(|A|) - \log(10c^2).\]
    Thus, we stop after at most $\log(|A|)- \log(10c^2)$ iterations. Since $c \geq 1$, this is at most $\log |A|-1$.
\end{proof}

\subsection{Runtime analysis}\label{runtime}
We get the following lemma from classical partition refinement techniques~\cite{PT87}.
See also~\cite{MPS} for a more direct graph-theoretic formulation of this result via modular decomposition.

\begin{lemma}[\cite{PT87,MPS}]\label{lemma:runtime-twin-classes}
    For a bipartite graph, $G = (A, B, E)$,
    the partition of \(A\) into twin classes over \(B\) can be computed in time \(\mathcal{O}(|A|+|B|+|E|)\).
\end{lemma}

The sampling process of \Cref{algorithm:welzl-orders} can be implemented, for example, using \emph{reservoir sampling}~\cite{vitter_random_1985}.
\begin{lemma}[\cite{vitter_random_1985}]\label{lem:samplingruntime}
    Given a set \(A\) and a target size \(s \le |A|\), one can sample \(W \subseteq A\) of size \(s\) uniformly at random in time \(\mathcal{O}(|A|)\).
\end{lemma}

\begin{lemma}\label{lemma:runtime-check-partition}
    Let $G = (A, B, E)$ be a bipartite graph, let $\mathcal{P} = \{P_1, \dots, P_\ell\}$ be a partition of $B$,
    and let \(P = \{v_1,\dots,v_\ell\}\) be a set of representatives of each partition class.
    In time \(\mathcal{O}(|A|+|B|+|E|)\), one can find the smallest \(k\) such that \(\mathcal P\) is a \(k\)-near twin partition of \(B\) with representatives \(P\), that is, one can compute \(k = \max_{i} \max_{v \in P_i} |N(v_i) \symdiff N(v)|\).
\end{lemma}
\begin{proof}
    First note that for any two vertices $v_1, v_2$ we have
    \[
        |N(v_1) \symdiff N(v_2)| = |N(v_1)| + |N(v_2)| - 2|N(v_1) \cap N(v_2)|.
    \]
    The neighborhood sizes of all vertices in \(B\) can be computed in linear time.
    To compute the intersections \(|N(v_i) \cap N(v)|\), we proceed as follows.
    Initialize a perfect hash map that lets us query, for each \(a \in A\) and part index \(i\), whether \(a \in N(v_i)\) in constant time.
    Initialize pointers that let us query, for each \(b \in B\), the index \(i\) with \(b \in P_i\) in constant time.
    Then, initialize a counter to zero for each \(b \in B\).
    Loop through all \((a,b) \in E\) and increment the counter at \(b\) if \(a \in N(v_i)\), where \(i\) is the index such that \(b \in P_i\).
    Using the above data structures, this takes constant time per edge.
    In the end, the counter at \(b \in P_i\) equals \(|N(v_i) \cap N(b)|\).
    Given these values, the maximum over all \(i\) and \(v \in P_i\) can be computed in \(\mathcal{O}(|B|)\) additional time.
\end{proof}

We are finally able to establish the runtime of our algorithm.

\begin{theorem}\label{theorem:runtime-bound}
    Given a set system $G = (A, B, E)$ of linearity~$c \ge 1$,
    \Cref{algorithm:welzl-orders} terminates in time $\mathcal{O}\bigl((|A| + |E|)\log |A|\bigr)$.
\end{theorem}
\begin{proof}
The initialization of the algorithm can be performed in time $\mathcal{O}(|A|+|B|)$.
By \Cref{theorem:iteration_bound}, the main loop executes at most $\log(|A|)$ times.
Each execution incurs the following costs:

\begin{itemize}
    \item \emph{Sampling.} By \Cref{lem:samplingruntime}, the sampling can be done in time \(\mathcal{O}(|A|)\). \smallskip

    \item \emph{Partitioning.}  As shown in \Cref{lemma:runtime-twin-classes}, the twin partitions \(\mathcal{A}\) and \(\mathcal{B}\) can be computed in time $\mathcal{O}(|A|+|B|+|E|)$.
        Representatives \(A'\) and \(B'\) can clearly be computed in time \(\mathcal{O}(|A|)\) and \(\mathcal{O}(|B|)\), respectively.

    \item \emph{Checking guarantees.} We can use \Cref{lemma:runtime-check-partition} with partition \(\mathcal B\), representatives \(B'\), and graph \(G[A_{\mathrm{cur}},B_{\mathrm{cur}}]\)
        to check in time \(\mathcal{O}(|A|+|B|+|E|)\) if
        the required guarantee is satisfied.
    \item \emph{Bookkeeping.}
        In time \(\mathcal{O}(|A|)\), we can set up a data structure that allows us to look up for each \(u \in A\) its representative \(u_{\mathrm{rep}} \in A'\) in constant time.
        As at most \(|A|\) pairs are pushed onto the stack, the total runtime is \(\mathcal{O}(|A|)\).

\end{itemize}

Afterwards, for the reconstruction,
we represent the order \(\prec\) as a doubly linked list with a hash table for \(\mathcal{O}(1)\) position lookup.
This allows insertions in \(\mathcal{O}(1)\) time.
Since \(|L| \le |A|\), the total reconstruction time is \(\mathcal{O}(|A|)\).
Summing over at most \(\log |A|\) iterations yields total time \(\mathcal{O}\bigl((|A| + |B| + |E|)\log |A|\bigr)\). We have $|B| \le |E|+1$, and thus the bound simplifies to $\mathcal{O}\bigl((|A| + |E|)\log |A|\bigr)$.
\end{proof}

\subsection{Bounding the crossing number}\label{sec:crossing-number}

Having established the bound on the number of iterations, it remains to analyze the resulting crossing number.

\begin{theorem}\label{theorem:resulting-crossing-number}
Given a set system $G = (A, B, E)$ of linearity~$c \ge 1$,
\Cref{algorithm:welzl-orders} either returns \emph{false} or computes an order on \(A\) with crossing number at most $12c^2\log^2 |A|$.
\end{theorem}
\begin{proof}
We analyze the crossing number of the output order assuming the algorithm does not return \emph{false}.
By the iteration bound in \Cref{theorem:iteration_bound}, the number $r$ of iterations of the main loop satisfies $r \le \log |A| - 1$.
For $1 \le i \le r$, let $(A_i, B_i)$ denote the value of $(A_{\mathrm{cur}}, B_{\mathrm{cur}})$ at the beginning of iteration $i$, as in the proof of \Cref{theorem:iteration_bound}, and let $(A_{r+1}, B_{r+1})$ denote its value when the main loop terminates.
Thus, $(A_1, B_1) = (A, B)$.
Let $\mathcal{B}_i$ be the partition of $B_i$ computed in iteration $i$ with representatives $B_{i+1}$.

The reconstruction phase builds an order $\prec$ on $A$ in $r+1$ steps, indexed $r+1,r,\ldots,1$.
We show by induction that after step $i$ the current order on $A_i$ has crossing number at most $(r - i + 2) \cdot 12c^2 \log |A|$ with respect to $G[A_i, B_i]$.

\textbf{Base case} ($i = r+1$):
The algorithm chooses an arbitrary order on $A_{r+1}$.
Since $|A_{r+1}| \le 12c^2 \log |A|$ by the loop termination condition
and any order on $n$ elements has crossing number at most $n - 1$,
the crossing number is at most $12c^2 \log |A|$.

\textbf{Inductive step} ($i \to i-1$):
Suppose we have an order on $A_i$ with crossing number at most $m = (r-i+2) \cdot 12c^2 \log |A|$ with respect to $G[A_i, B_i]$.
The reconstruction inserts each vertex $a \in A_{i-1} \setminus A_i$ immediately after its representative in $A_i$.
Since $A_i$ consists of representatives of the twin partition of $A_{i-1}$ over $B_i$,
each inserted vertex $a$ is a twin of its representative with respect to $B_i$.
By \Cref{lemma:suspend-twins}, these insertions do not change the crossing number with respect to $G[A_{i-1}, B_i]$.

It remains to account for the vertices in $B_{i-1} \setminus B_i$.
Since the algorithm did not return \emph{false} in iteration $i-1$,
$\mathcal{B}_{i-1}$ is a $(6c^2 \log |A|)$-near twin partition of $B_{i-1}$ with representatives $B_i$.
By \Cref{lemma:2k-increment} with $k = 6c^2 \log |A|$,
the crossing number with respect to $G[A_{i-1}, B_{i-1}]$ is at most $m + 2k = m + 12c^2 \log |A|$.
This equals $(r - (i-1) + 2) \cdot 12c^2 \log |A|$, completing the induction.

After all removed vertices have been reinserted, the reconstructed order on the original set $A_1 = A$ has crossing number at most
\[
    (r+1) \cdot 12c^2 \log |A| \le \log |A| \cdot 12c^2 \log |A| = 12c^2 \log^2 |A|. \qedhere
\]
\end{proof}

We next bound the probability that a run returns \emph{false}.

\subsection{Bounding the failure probability}\label{section:probability_of_bad_run}
Intuitively, a sample is bad if there exist $b,b'\in B$ with a large symmetric difference $|N(b)\symdiff N(b')| \geq 6c^2 \log N$ such that the sample $W$ contains no vertex from this symmetric difference, and therefore fails to distinguish $b$ and $b'$.

\begin{lemma}\label{lemma:missing_large_set_probability}
    Let $G=(A,B,E)$ be a bipartite graph, let $c \ge 1$ be a constant,
    and let $X \subseteq A$ with $|X|\ge 6c^2 \log N$ for some \(N \geq |A|\).
    For $W \subseteq A$ sampled uniformly at random of size $|W| = \left\lceil |A|/(2c^2) \right\rceil$, we have
    \[
        \Pr\!\left[\,X\cap W=\emptyset\,\right]\ \le\ N^{-3}.
    \]
\end{lemma}
\begin{proof}
    Since $W$ is chosen uniformly among all subsets of this size, we have
    \[
        \Pr[X \cap W = \emptyset]
        \;=\;
        \frac{\binom{|A| - |X|}{|W|}}{\binom{|A|}{|W|}}
        \;=\;
        \prod_{i=0}^{|W|-1} \frac{|A| - |X| - i}{\,|A| - i\,}
        \;\le\;
        \left(1 - \frac{|X|}{|A|}\right)^{\!|W|}
        \;\le\;
        \exp\!\left(-\frac{|W|}{|A|}|X|\right).
    \]
    With $|W| = \left\lceil |A|/(2c^2) \right\rceil$ and $|X| \ge 6c^2 \log N$, we obtain
    \begin{align*}
        \Pr[X \cap W = \emptyset]
        \;\le\;
        \exp\!\left(-\frac{|W|}{|A|}|X|\right)
        \;&\leq\;
        \exp\!\left(-\frac{1}{2c^2}\cdot 6c^2 \log N\right)
        \;\\&=\;
        \exp\!\left(- 3 \log N\right)
        \;\\&=\;
        N^{\tfrac{-3}{\ln 2}},
        \;\\&\leq\;
        N^{-3},
    \end{align*}
    which proves the claim.
\end{proof}

\begin{lemma}\label{lemma:twin_class_occurrence_probability}
    Let $G = (A, B, E)$ be a bipartite graph of linearity~$c \ge 1$ and let \(N \geq |A|\).
    For a uniformly random sample \(W \subseteq A\) of
    size $|W| = \lceil |A|/2c^2 \rceil$,
    the probability that there are $b, b' \in B$ with \(N(b) \cap W = N(b') \cap W\) and
    $|N(b) \symdiff N(b')| \ge 6c^2 \log N$ is at most \(c^2/N\).
\end{lemma}
\begin{proof}
As the graph has linearity \(c\), the vertices in \(B\) have at most $c \cdot |A|$ distinct neighborhoods in \(A\).
To prove this lemma, it is sufficient to restrict the choices of \(b,b'\) to representatives of the $\le c \cdot |A|$ twin classes of \(B\).
We can hence assume without loss of generality that \(|B| \le c \cdot |A|\).
Fix $b, b' \in B$.
By \Cref{lemma:missing_large_set_probability} with $X = N_G(b) \symdiff N_G(b')$ we have
\[
\Pr\bigl[|X| \ge 6c^2 \log N \textnormal{~and~} X \cap W = \emptyset\,\bigr] \le N^{-3}.
\]

There are $\binom{|B|}{2} < |B|^2 \le c^2 |A|^2$ unordered pairs $b,b' \in B$ that could satisfy this condition. 
By the union bound, the probability that \emph{any} such pair satisfies it is at most
\(c^2 |A|^2/N^3 \leq c^2/N\), proving the lemma.
\end{proof}

We now bound the probability that the algorithm returns \emph{false}.
Recall that the algorithm returns \emph{false} only when the guarantee check fails,
which happens when the sampled partition $\mathcal{B}$ is not a $(6c^2 \log |A|)$-near twin partition.

\begin{theorem}\label{theorem:probability_of_error}
    \Cref{algorithm:welzl-orders} returns \emph{false} with probability at most \(1/3\).
\end{theorem}
\begin{proof}
    Let $G = (A, B, E)$ be the input graph.
    In a given iteration of the main loop, let $G[A_{\mathrm{cur}}, B_{\mathrm{cur}}]$ be the current graph and \(W \subseteq A_{\mathrm{cur}}\) the sample.
    We call \(W\) \emph{bad} if for the resulting partition $\mathcal{B}$ there exist $b, b' \in B_{\mathrm{cur}}$ in the same part of \(\mathcal{B}\) with
    \[
        |(N(b) \cap A_{\mathrm{cur}}) \symdiff (N(b') \cap A_{\mathrm{cur}})| > 6c^2 \log |A|.
    \]
    The algorithm only fails the guarantee check and returns \emph{false} if \(W\) is bad.
    By \Cref{lemma:twin_class_occurrence_probability} applied to the current graph with \(N = |A|\), the probability that \(W\) is bad is at most $c^2 / |A|$.

    By \Cref{theorem:iteration_bound}, there are at most $\log |A|$ iterations.
    By the union bound, the probability that any iteration has a bad sample is at most \(\log |A| \cdot c^2/|A|\).
    The main loop only executes when $|A_{\mathrm{cur}}| > 12c^2 \log |A|$,
    so in particular $|A| \ge 12c^2 \log |A|$.
    Therefore,
    \[
        \log |A| \cdot \frac{c^2}{|A|}
        \le \frac{c^2 \log |A|}{12c^2 \log |A|}
        = \frac{1}{12}
        < \frac{1}{3}. \qedhere
    \]
\end{proof}

Having established the guarantees of the uniform-sampling algorithm,
we now address the source of the second logarithmic factor in its crossing bound:
the $\mathcal O_c(\log |A|)$ difference between sets that agree on the sample.
The refined sampling variant bounds this difference independently of $|A|$,
giving the $\mathcal O_c(\log |A|)$ crossing guarantee of \Cref{thm:refined-main}
while preserving the partitioning and reconstruction of the simple algorithm.


\section{Refined sampling and logarithmic crossing number}
\label{sec:refined-sampling}

We improve the crossing-number bound by replacing the uniform sample in
\Cref{algorithm:welzl-orders} with an augmented sample satisfying the same
size bound.
The procedure first samples vertices uniformly, identifies vertices that
are not well represented by the sample, and augments the sample until all
remaining twin classes are indeed near twin classes. The resulting
near-twin bound depends only on $c$; the partitioning, bookkeeping,
and reconstruction steps of the outer algorithm remain unchanged.
The refined crossing bound depends on $c$ as
$\widetilde{\mathcal O}(c^5)$, while the expected time for constructing a
sample is linear in the input size. Here $\widetilde{\mathcal O}$ hides
polylogarithmic factors in $c$.

We prove the refined guarantee stated in \Cref{thm:refined-main} by first
constructing the sample and then substituting it into the outer algorithm.

The sampling analysis builds heavily on the work
in~\cite[Lemma~15 in Appendix~7.3]{KZ20}
and~\cite[Lemmas~25 and~26]{BDPST26}. We obtain the following sampling
guarantee.

\begin{lemma}[Refined sampling]
\label{lem:refined-sampling}
There is a randomized algorithm that, given $c\geq2$ and a set system
$G=(A,B,E)$ of linearity at most $c$, with $A\neq\varnothing$, computes in
expected time $\mathcal O(|A|+|B|+|E|)$
a set $W\subseteq A$ of size
\[
    |W|\leq\left\lceil\frac{|A|}{2c^2}\right\rceil
\]
such that, for all $b,b'\in B$,
\[
    N(b)\cap W=N(b')\cap W
    \quad\Longrightarrow\quad
    |N(b)\symdiff N(b')|=\mathcal O(c^5\log^2 c)
    =\widetilde{\mathcal O}(c^5).
\]
\end{lemma}

In other words, the twin partition of $B$ in $G[W,B]$ is $k$-near
in $G$ for any choice of representatives, with
$k=\widetilde{\mathcal O}(c^5)$.

\subsection{Construction of the sample}
Set
\[
    s=\lceil64\log c\rceil,\qquad
    r=\left\lfloor\frac{|A|}{4c^2}\right\rfloor,\qquad
    k=\lceil768c^5(s+1)\log c\rceil.
\]
If $|A|\leq k$ or $B=\varnothing$, return an arbitrary subset of $A$ of
size at most $\lceil|A|/(2c^2)\rceil$. Otherwise, as in the proof of
\Cref{lemma:twin_class_occurrence_probability}, it suffices to restrict $B$
to one representative of each twin class over $A$. We can hence assume
without loss of generality that the vertices in $B$ have distinct neighborhoods.

\begin{enumerate}
\item \textbf{Initial sampling.}
Take $r$ independent uniform draws from $A$, with replacement, and let
$W$ be their underlying set.
Compute the twin partition $\mathcal B=\{B_1,\ldots,B_\ell\}$ of $B$
in $G[W,B]$. Independently choose a representative $b_i$ uniformly from
each twin class $B_i$, and let $B'=\{b_1,\ldots,b_\ell\}$.
Using neighborhoods in $G$, we call a vertex $b\in B_i$ \emph{exceptional} if
\[
    |N(b)\symdiff N(b_i)|>k/3.
\]
Let $q$ be the number of exceptional vertices. If
\[
    q>\frac{|A|}{4c^2(s+1)},
\]
discard the attempt and repeat both the initial sampling and the choice of
representatives.

\item \textbf{Distinguishing the exceptions.}
We now choose vertices in $A$ to distinguish the exceptional vertices in $B$
from one another.
Process the vertices $a\in A$ one at a time. If two exceptional vertices
have equal neighborhoods on the current $W$ and $a$ is adjacent to exactly
one of them, add $a$ to $W$.

Recompute the twin classes of $B$ over the enlarged $W$.
Each class contains at most one exceptional vertex, since $W$
distinguishes these vertices from one another.

\item \textbf{Local repairs.}
We now distinguish exceptional vertices from the nonexceptional vertices in their twin classes.
For each exceptional vertex $b$, let $B_i$ be its twin class from Step~1, with representative $b_i$.
If $|N(b)\symdiff N(b_i)|\leq2k/3$, no local repair is necessary.
Otherwise let $X=N(b)\symdiff N(b_i)$ and draw $s$ elements independently and uniformly
from $X$, with replacement.
Check whether their underlying set distinguishes $b$ from every other
vertex in its twin class computed in Step~2. If the trial fails, discard all of its draws and repeat;
otherwise add the underlying set of a successful trial to $W$.
\end{enumerate}

Return $W$.

\begin{claim*}[Sample size]
Whenever the procedure terminates, it returns a set $W$ of size
at most $\lceil|A|/(2c^2)\rceil$.
\end{claim*}

\begin{proof}
In the accepted attempt, each vertex added in Step~2 increases the number
of twin classes among the $q$ exceptions. Since these vertices form at most
$q$ twin classes and classes never merge, this can happen at most $q$ times,
so Step~2 adds at most $q$ vertices.
Step~3 adds at most $sq$ vertices, at most $s$ for each exceptional vertex. Hence
\[
    |W|\leq r+(s+1)q
    \leq\left\lfloor\frac{|A|}{4c^2}\right\rfloor+\frac{|A|}{4c^2}
    \leq\frac{|A|}{2c^2}.\qedhere
\]
\end{proof}

\begin{claim*}[Correctness]
Whenever the procedure terminates, any two vertices $b,b'\in B$ with
$N(b)\cap W=N(b')\cap W$ satisfy $|N(b)\symdiff N(b')|\leq k$.
\end{claim*}

\begin{proof}
Every final class lies within an initial twin class and thus has a common
initial representative. After Step~2, every twin class in our partition
contains at most one exceptional vertex. Any exception has
distance at most $2k/3$ from the initial representative, since otherwise
its repair isolates it; all nonexceptional vertices have distance at
most $k/3$. The triangle inequality therefore bounds every pairwise
distance by $k$. Restoring suppressed
twins preserves this bound.
\end{proof}

\begin{claim*}[Expected number of initial attempts]
The expected number of initial attempts is at most two.
\end{claim*}

\begin{proof}
If $|A|\leq k$ or $B=\varnothing$, the procedure returns immediately.
Otherwise, as in the proof of \Cref{lemma:twin_class_occurrence_probability},
we have retained one representative of each twin class of $B$ over $A$,
so $|B|\leq c|A|$. We show that an initial attempt is rejected with
probability at most $1/2$, which gives the claimed bound on the expected
number of attempts.

For one initial attempt, let $Z$ be the total distance from the vertices
in $B$ to their representatives:
\[
    Z=\sum_{i=1}^{\ell}\sum_{v\in B_i}|N(v)\symdiff N(b_i)|.
\]
By definition, each of the $q$ exceptional vertices contributes more than
$k/3$ to $Z$, so $q\leq3Z/k$ holds in every attempt. To bound
$\mathbb E q$, it therefore suffices to bound $\mathbb E Z$.

To estimate $\mathbb E Z$, consider $r+1$ independent uniform draws from
$A$ and delete a uniformly chosen position. The remaining $r$ draws have
the distribution used in Step~1 and are independent of the deleted vertex.
Choose representatives uniformly in the twin classes on these remaining
draws, also as in Step~1. Since the deleted vertex is uniform in $A$, the
expected number of vertices whose adjacency to it differs from that of
their representative is $\mathbb E Z/|A|$.

Fix the $r+1$ draws and consider the twin partition of $B$ on their
positions, keeping repeated draws as separate coordinates. Form a graph
whose vertices are these twin classes, with two classes adjacent when
their neighborhoods on the sample positions differ in exactly one
coordinate. The neighborhood system has VC-dimension
$d\leq4\log c$ by the bound in~\cite[proof of Theorem~32]{BDPST26},
and restriction to the sample positions preserves this bound: repeated
coordinates cannot be shattered independently. By the one-inclusion
graph bound in~\cite[Lemma~15 in Appendix~7.3]{KZ20}, this graph therefore
has an orientation with out-degree at most $d$.
This lets us bound the number of disagreements with the representatives
by assigning each edge to its tail, with at most $d$ edges assigned to any
twin class.

Deleting a position merges exactly those pairs of twin classes joined by
an edge corresponding to that position. If the two classes have sizes
$x$ and $y$, a uniformly chosen representative of their union disagrees
with an expected $2xy/(x+y)\leq2\min(x,y)$ vertices at the deleted position.
Charge this contribution to the tail of the oriented edge, whose class
size is at least $\min(x,y)$. Summing over all $r+1$ positions counts each
edge once, and each class receives at most $d$ charges, each at most
twice its size. The total is therefore at most $2d|B|$.

Averaging over the deleted position and then over the draws gives
\[
    \frac{\mathbb E Z}{|A|}\leq\frac{2d|B|}{r+1},
    \qquad\text{and hence}\qquad
    \mathbb E Z\leq\frac{8|A||B|\log c}{r+1}
    \leq32c^3|A|\log c,
\]
where the last inequality uses $|B|\leq c|A|$ and
$r+1>|A|/(4c^2)$. Together with $q\leq3Z/k$ and the choice of $k$, this yields
\[
    \mathbb E q\leq\frac{96c^3|A|\log c}{k}
    \leq\frac{|A|}{8c^2(s+1)}.
\]
Thus, the expected number of exceptional vertices is at most half the
acceptance threshold in Step~1.

An initial attempt is rejected only when $q>|A|/(4c^2(s+1))$.
By Markov's inequality and the preceding bound, its rejection probability is
\[
    \Pr\!\left(q>\frac{|A|}{4c^2(s+1)}\right)\leq\frac12.
\]
Since the sample and the representatives are chosen independently anew
in every attempt, the expected number of attempts is at most two, and
Step~1 terminates almost surely.
\end{proof}

\begin{claim*}[Expected number of local repair trials]
Each local repair requires at most two trials in expectation.
\end{claim*}

\begin{proof}
Fix an exceptional vertex $b\in B_i$ that requires a local repair, and
let $B_i'$ be its twin class from Step~2. We show that a trial fails to
distinguish $b$ from the other vertices of $B_i'$ with probability at most
$1/2$, which gives the claimed bound on the expected number of trials.

The trial draws from $X=N(b)\symdiff N(b_i)$, where $|X|>2k/3$.
By Step~2, every other vertex $v\in B_i'$ is nonexceptional and belongs
to the same initial twin class $B_i$, so
$|N(v)\symdiff N(b_i)|\leq k/3$. The vertices in $X$ that distinguish
$b$ from $v$ form the set
\[
    D_v=X\cap(N(b)\symdiff N(v)),
    \qquad |D_v|\geq|X|-k/3>|X|/2.
\]
Thus, more than half the vertices in $X$ distinguish $b$ from any fixed
vertex $v\in B_i'\setminus\{b\}$.

To bound the probability that the sample distinguishes $b$ from all of
$B_i'\setminus\{b\}$, consider the set system
$\{D_v:v\in B_i'\setminus\{b\}\}$ on $X$. For any $Y\subseteq X$ of
size $t$, we have
\[
    D_v\cap Y=(N(b)\cap Y)\symdiff(N(v)\cap Y).
\]
Since $N(b)\cap Y$ is fixed, the number of distinct intersections
$D_v\cap Y$ is at most the number of twin classes of $B$ in $G[Y,B]$,
which is at most $ct$ by the linearity of $G$. Hence, this set system
has shatter function at most $ct$.

A trial succeeds precisely when its sample meets every set $D_v$.
Since every such set has size greater than $|X|/2$, an $\varepsilon$-net
with $\varepsilon=1/2$ suffices. By the $\varepsilon$-net sampling bound
in~\cite[proofs of Lemmas~3.4 and~3.5]{HW87}, using $s\geq16$, we obtain
\[
    \Pr[\text{trial fails}]\leq4cs\,2^{-s/4}<\frac12,
\]
where the last inequality follows from $s=\lceil64\log c\rceil$ and
$c\geq2$. Since each trial draws independently anew from the fixed set
$X$, the expected number of trials is at most two, and the local repair
terminates almost surely.
\end{proof}

We now show that all local repair trials can be checked in expected linear
time, which we use in the running-time analysis.

\begin{lemma}[Trial checking]
\label{lem:refined-trial-checking}
All Step~3 checks take expected time $\mathcal O(|A|+|B|+|E|)$.
\end{lemma}

\begin{proof}
For a trial repairing $b$ in its class $B_i'$ from Step~2, let $S$ be the
underlying set of its draws. As in the proof of
\Cref{lemma:runtime-check-partition}, compute the sizes
$|S\cap(N(b)\symdiff N(v))|$ by marking $S$ and $S\cap N(b)$ in two arrays
on $A$ and scanning each neighborhood $N(v)$, for
$v\in B_i'\setminus\{b\}$. The trial succeeds precisely when all these
sizes are positive. Initialize the arrays once and clear only marked
entries after each trial, so each check takes time
\[
    \mathcal O\!\left(s+|B_i'|+\sum_{v\in B_i'}|N(v)|\right).
\]

To sum these costs, note that the classes from Step~2 are disjoint and
contain at most one exceptional vertex each. By the preceding trial bound,
each repair requires at most two checks in expectation. Including the
initialization, the total expected checking time is therefore
\[
    \mathcal O(|A|+sq+|B|+|E|)=\mathcal O(|A|+|B|+|E|),
\]
since $(s+1)q\leq|A|/(4c^2)$.
\end{proof}

\begin{claim*}[Expected running time]
The procedure runs in expected time
$\mathcal O(|A|+|B|+|E|)$.
\end{claim*}

\begin{proof}
The initialization computes the twin classes of $B$ over $A$ and retains
one representative per class in time $\mathcal O(|A|+|B|+|E|)$ by
\Cref{lemma:runtime-twin-classes}. The remaining costs are:

\begin{itemize}
\item \emph{Initial sampling.}
Taking the draws, computing the initial twin classes, choosing their
representatives, and identifying the exceptions takes time
$\mathcal O(|A|+|B|+|E|)$ by
\Cref{lemma:runtime-twin-classes,lemma:runtime-check-partition}.
At most two attempts are needed in expectation, giving the same expected
time bound.

\item \emph{Distinguishing the exceptions.}
Apply the partition refinement from \Cref{lemma:runtime-twin-classes}
to the exceptions, adding $a$ to $W$ whenever adjacency to $a$ splits
a class. Processing only adjacent vertices takes constant work per
vertex and edge, so this pass and recomputing the twin partition of $B$
take time $\mathcal O(|A|+|B|+|E|)$.

\item \emph{Local repairs.}
Construct each set $X$ once and reuse it for all trials. Each exception's
neighborhood is scanned once. For any fixed initial partition, a
representative's neighborhood is scanned at most once per vertex in its
class. Uniform choice therefore gives
\[
    \mathbb E\!\left[\sum_{i=1}^{\ell}|B_i|\,|N(b_i)|\right]
    =\sum_{v\in B}|N(v)|\leq|E|.
\]
Conditioning on acceptance costs at most a factor two, so all sets $X$
are constructed in expected time $\mathcal O(|A|+|B|+|E|)$.

By \Cref{lem:refined-trial-checking}, all trial checks take expected time
$\mathcal O(|A|+|B|+|E|)$. Drawing the samples and adding successful sets
costs $\mathcal O(s)$ per trial, which is absorbed by the same bound.
\end{itemize}

Summing these costs gives expected time
$\mathcal O(|A|+|B|+|E|)$.
\end{proof}

Together, these claims prove \Cref{lem:refined-sampling}.

\subsection{Substitution in
\texorpdfstring{\Cref{algorithm:welzl-orders}}{Algorithm~\getrefnumber{algorithm:welzl-orders}}}
In \Cref{algorithm:welzl-orders}, replace the sampling step by
\Cref{lem:refined-sampling}, omit the guarantee check, and replace the
loop condition by $|A_{\mathrm{cur}}|>2k$, with $k$ as above.
All remaining steps are unchanged.

\begin{theorem}[Logarithmic crossing number]
\label{thm:refined-welzl}
Let $G=(A,B,E)$ be a set system of linearity at most $c\geq2$, with
$|A|\geq2$. \Cref{algorithm:welzl-orders}, with the
substitutions above, is a randomized algorithm that computes an
order on $A$ with crossing number at most
\[
    \mathcal O\bigl(c^5\log^2 c\,\log|A|\bigr)
    =\mathcal O_c(\log|A|)
\]
in expected time
$\mathcal O\bigl((|A|+|B|+|E|)\log|A|\bigr)$.
\end{theorem}

\begin{proof}
The sample satisfies the same size upper bound, so the recurrence in the proof of
\Cref{theorem:iteration_bound} still applies. Since $2k\geq12c^2$, it gives
at most $\log|A|-1$ iterations with the new stopping condition.

Each sampled twin partition is $k$-near, so each reconstruction step adds
at most $2k$ crossings by \Cref{lemma:suspend-twins,lemma:2k-increment}.
The base order has at most $2k$ crossings by the stopping condition, giving
a final crossing number at most
\[
    2k\log|A|=\mathcal O\bigl(c^5\log^2 c\,\log|A|\bigr).
\]

Each iteration has expected cost
$\mathcal O(|A|+|B|+|E|)$ by
\Cref{lem:refined-sampling} and the proof of \Cref{theorem:runtime-bound}.
Summing over the iterations and adding the linear reconstruction cost gives
the expected running time.
In the canonical set-system representation, $|B|\leq|E|+1$ and
$|A|+|E|=\|S\|$, so this also proves \Cref{thm:refined-main}.
\end{proof}

\paragraph*{Use of generative AI}
AI assistance was used for writing and editing this paper. All proofs and
arguments are our own, with one exception: an algorithmic improvement that
strengthens the crossing-number bound from $\mathcal O(\log^2 n)$ to
$\mathcal O(\log n)$ was suggested by AI. We take full responsibility for
the correctness and presentation of this work.

\bibliography{icalp_refined}
\bibliographystyle{plainurl}
\end{document}